\documentclass{amsart}

\usepackage{graphicx}          
                               
\include{diagrams}
\usepackage{diagrams}
\usepackage{tikz}
\usetikzlibrary{automata}
\usepackage{subfigure}
\usepackage{amssymb,latexsym,amsfonts,amsmath}
\usepackage{graphicx}
\newtheorem{theorem}{Theorem}[section]
\newtheorem{lemma}[theorem]{Lemma}
\newtheorem{problem}[theorem]{Problem}

\newtheorem{definition}[theorem]{Definition}

\newtheorem{remark}[theorem]{Remark}

\newcommand{\Post}{\operatorname{Post}}
\newcommand{\delay}{\mathrm{delay}}
\newcommand{\ctrl}{\mathrm{ctrl}}
\newcommand{\send}{\mathrm{send}}
\newcommand{\scc}{\mathrm{sc}}
\newcommand{\ca}{\mathrm{ca}}
\newcommand{\req}{\mathrm{req}}
\newcommand{\alt}{\mathrm{alt}}

\begin{document}

\title[A Symbolic Approach to the Design of \\ Nonlinear Networked Control Systems]{A Symbolic Approach to the Design of \\ Nonlinear Networked Control Systems}
\thanks{The research leading to these results has been partially supported by the Center of Excellence DEWS and received funding from the European Union Seventh Framework Programme [FP7/2007-2013] under grant agreement n. 257462 HYCON2 Network of excellence.}

\author[Alessandro Borri, Giordano Pola, Maria D. Di Benedetto]{
Alessandro Borri$^{\ast}$, Giordano Pola$^{\ast}$, Maria D. Di Benedetto$^{\ast}$}
\address{$^{\ast}$
Department of Electrical and Information Engineering, Center of Excellence DEWS,
University of L{'}Aquila, 67100 L{'}Aquila, Italy}
\email{ \{alessandro.borri,giordano.pola,mariadomenica.dibenedetto\}@univaq.it}

\begin{abstract}
Networked control systems (NCS) are spatially distributed systems where communication among plants, sensors, actuators and controllers occurs in a shared communication network. NCS have been studied for the last ten years and important research results have been obtained. These results are in the area of stability and stabilizability. However, while important, these results must be complemented in different areas to be able to design effective NCS. In this paper we approach the control design of NCS using symbolic (finite) models. Symbolic models are abstract descriptions of continuous systems where one symbol corresponds to an "aggregate" of continuous states. 
We consider a fairly general multiple-loop network architecture where plants communicate with digital controllers through a shared, non-ideal, communication network characterized by variable sampling and transmission intervals, variable communication delays, quantization errors, packet losses and limited bandwidth. We first derive a procedure to obtain symbolic models that are proven to approximate NCS in the sense of alternating approximate bisimulation. We then use these symbolic models to design symbolic controllers that realize specifications expressed in terms of automata on infinite strings. An example is provided where we address the control design of a pair of nonlinear control systems sharing a common communication network. The closed--loop NCS obtained is validated through the OMNeT++ network simulation framework.
\end{abstract}

\maketitle



\section{Introduction}
In the last decade, the integration of physical processes with networked computing units led to a new generation of control systems, termed Networked Control Systems (NCS). NCS are complex, heterogeneous, spatially distributed systems where physical processes interact with distributed computing units through non--ideal communication networks. While the process is often described by continuous dynamics, algorithms implemented on microprocessors in the computing units are generally modeled by finite state machines or other models of computation. In addition, communication network properties depend on the features of the communication channel and of the protocol selected, e.g. sharing rules and wired versus wireless network. In the last few years NCS have been the object of great interest in the research community and important research results have been obtained with respect to stability and stabilizability problems, see e.g. \cite{NCSHesphana,HeemelsSurvey,NCS-HSCC2010}. However, these results must be complemented to meet more general and complex specifications when controlling a NCS. In this paper, we propose to approach the control design of NCS by using symbolic (finite) models (see e.g. \cite{DiscAbs,paulo} and the references therein), which are typically used to address control problems where software and hardware interact with the physical world. 

This paper presents two connected results. The first is a novel approach to NCS modeling, where a wide class of non-idealities in the communication network are considered such as variable sampling/transmission intervals, variable communication delays, quantization errors, packet dropouts and limited bandwidth. By using this general approach to modeling a NCS, we can derive symbolic models that approximate incrementally stable \cite{IncrementalS} nonlinear NCS in the sense of alternating approximate bisimulation \cite{PolaSIAM2009} with arbitrarily good accuracy. This result is strong since the existence of an alternating approximate bisimulation guarantees that (i) control strategies synthesized on the symbolic models can be applied to the original NCS, independently of the particular realization of the non--idealities in the communication network; (ii) if a solution does not exist for the given control problem (with desired accuracy) for the symbolic model, no control strategy exists for the original NCS. 
The second result is about the design of a NCS where the control specifications are expressed in terms of automata on infinite strings. Given a NCS and a specification, we explicitly derive a symbolic controller such that the controlled system meets the specification \emph{in the presence of the considered non-idealities in the communication network}. To illustrate the use of our results, we apply the methodology to derive a controller for a pair of nonlinear systems sharing a common communication network. To validate the controller, the closed--loop NCS is simulated in the OMNeT++ network simulation framework \cite{Omnet}. The results of this paper follow the approach on construction of symbolic models for nonlinear control systems reported in \cite{PolaAutom2008,PolaSIAM2009,PolaSCL10,PolaTAC12,MajidTAC11}. \\
The paper is organized as follows. Section 2 introduces the notation employed in the sequel. In Section 3 we present the class of networked control systems that we consider in the paper. Section 4 reports some preliminary definitions of the notions of systems, approximate bisimulation and approximate parallel composition. Section 5 proposes symbolic models that approximate incrementally stable NCS in the sense of alternating approximately bisimulation. In Section 6 we address the symbolic control design of NCS. A realistic implementation of the symbolic control of a NCS on \mbox{OMNeT++} is included in Section 7. Section 8 offers concluding remarks.

\section{Notation}\label{sec:Notation}
The identity map on a set $A$ is denoted by $1_{A}$. Given two sets $A$ and $B$, if $A$ is a subset of $B$ we denote by $1_{A}:A\hookrightarrow B$ or simply by $\imath$ the natural inclusion map taking any $a\in A$ to $\imath (a)  =a\in B$. 
Given a set $A$ we denote $A^{2}=A\times A$ and $A^{n+1}=A\times A^{n}$ for any $n\in \mathbb{N}$. 
Given a pair of sets $A$ and $B$ and a function $f:A\rightarrow B$ we denote by $f^{-1}: B\rightarrow A$ the inverse function of $f$ such that $f^{-1}(b)=a$ if and only if $f(a)=b$ for any $a\in A$. 
Given a pair of sets $A$ and $B$ and a relation $\mathcal{R}\subseteq A\times B$, the symbol $\mathcal{R}^{-1}$ denotes the inverse relation of $\mathcal{R}$, i.e. $\mathcal{R}^{-1}:=\{(b,a)\in B\times A:( a,b)\in \mathcal{R}\}$. The symbols $\mathbb{N}$, $\mathbb{N}_0$, $\mathbb{Z}$, $\mathbb{R}$, $\mathbb{R}^{+}$ and $\mathbb{R}_{0}^{+}$ denote the set of natural, nonnegative integer, integer, real, positive real, and nonnegative real numbers, respectively. Given an interval $[a,b]\subseteq \mathbb{R}$ with $a\leq b$ we denote by $[a;b]$ the set $[a,b]\cap \mathbb{N}$. We denote by $\lfloor x\rfloor:=\max\{{n\in\mathbb{Z} \vert n\leq x}\}$ the floor and by $\lceil x \rceil:=\min\{{n\in\mathbb{Z} \vert n\geq x}\}$ the ceiling of a real number $x$. Given a vector $x\in\mathbb{R}^{n}$ we denote by $\Vert x\Vert$ the infinity norm and by $\Vert x\Vert_{2}$ the Euclidean norm of $x$. A continuous function \mbox{$\gamma:\mathbb{R}_{0}^{+}\rightarrow\mathbb{R}_{0}^{+}$} is said to belong to class $\mathcal{K}$ if it is strictly increasing and
\mbox{$\gamma(0)=0$}; a function $\gamma$ is said to belong to class $\mathcal{K}_{\infty}$ if \mbox{$\gamma\in\mathcal{K}$} and $\gamma(r)\rightarrow\infty$
as $r\rightarrow\infty$. A continuous function \mbox{$\beta:\mathbb{R}_{0}^{+}\times\mathbb{R}_{0}^{+}\rightarrow\mathbb{R}_{0}^{+}$} is said to belong to class $\mathcal{KL}$ if for each fixed $s$ the map $\beta(r,s)$ belongs to class $\mathcal{K}_{\infty}$ with respect to $r$ and for each fixed $r$ the map $\beta(r,s)$ is decreasing with respect to $s$ and $\beta(r,s)\rightarrow0$ as \mbox{$s\rightarrow\infty$}. 
Given $\mu\in\mathbb{R}^{+}$ and $A\subseteq \mathbb{R}^{n}$, we set $[A]_{\mu}=\mu\mathbb{Z}^{n} \cap A$; if $B=\bigcup_{i\in [1;N]}A^{i}$ then $[B]_{\mu}=\bigcup_{i\in [1;N]} ([A]_{\mu})^{i}$. Consider a bounded set $A \subseteq \mathbb{R}^n$ with interior. Let $H=[a_1,b_1]\times[a_2,b_2]\times \dots \times [a_n,b_n]$ be the smallest hyperrectangle containing $A$ and set $\hat{\mu}_{A}=\min_{i=1,2,\dots,n} (b_i-a_i)$. It is readily seen that for any $\mu \leq \hat{\mu}_A$ and any $a\in A$ there always exists $b\in [A]_{\mu}$ such that $\Vert a-b \Vert \leq \mu$. 
Given $a\in A\subseteq \mathbb{R}^{n}$ and a precision $\mu\in\mathbb{R}^{+}$, the symbol $[a]_{\mu}$ denotes a vector in $\mu \, \mathbb{Z}^{n}$ such that $\Vert a-[a]_{\mu} \Vert \leq \mu/2$. Any vector $[a]_{\mu}$ with $a\in A$ can be encoded by a finite binary word of length 
$ \lceil\log_{2} \vert [A]_{\mu} \vert \rceil$.  

\section{Networked Control Systems}\label{sec:modelingNCS}
The class of Network Control Systems (NCS) that we consider in this paper has been inspired by the models reviewed in \cite{HeemelsSurvey} and is depicted in Figure \ref{NCSpic}. The sub--systems composing the NCS are described hereafter.\\

\begin{figure}
\begin{center}
\includegraphics[scale=1.2]{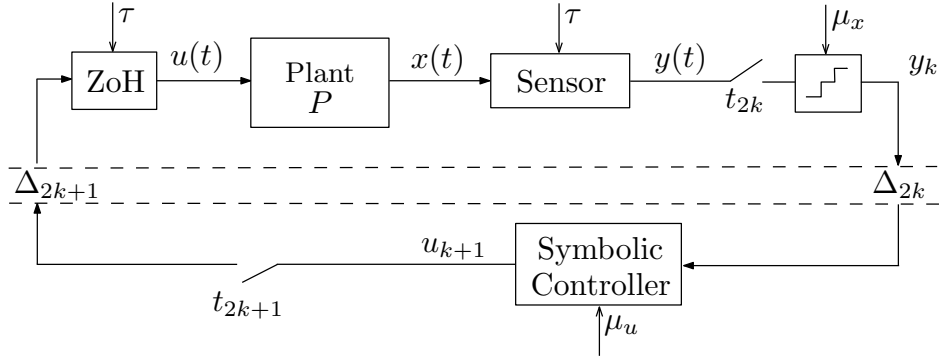}
\caption{Networked control system.}
\label{NCSpic}
\end{center}
\end{figure}

\textit{Plant.} The plant $P$ of the NCS is a nonlinear control system in the form of: 
\begin{equation}
\left\{
\begin{array}
{l}
\dot{x}(t)=f(x(t),u(t)),\\
x\in X\subseteq \mathbb{R}^{n},\\
x(0)\in X_{0}\subseteq X,\\
u(\cdot)\in \mathcal{U},
\end{array}
\right.
\label{NCSeq1}
\end{equation}
where $x(t)$ and $u(t)$ are the state and the control input at time $t\in\mathbb{R}^{+}_{0}$, $X$ is the state space, $X_{0}$ is the set of initial states and $\mathcal{U}$ is the set of control inputs that are supposed to be piecewise--constant functions of time from intervals of the form \mbox{$]a,b[\subseteq\mathbb{R}$} to $U\subseteq \mathbb{R}^{m}$. The set $U$ is assumed to be compact, convex with the origin as an interior point. The function $f:X\times U \rightarrow X$ is such that $f(0,0)=0$ and assumed to be Lipschitz on compact sets. 
In the sequel we denote by $\mathbf{x}(t,x_{0},u)$ the state reached by (\ref{NCSeq1}) at time $t$ under the control input $u$ from the initial state $x_{0}$; this point is uniquely determined, since the assumptions on $f$ ensure existence and uniqueness of trajectories. We assume that the control system $P$ is forward complete, namely that every trajectory is defined on an interval of the form $]a,\infty\lbrack$. Sufficient and necessary conditions for a control system to be forward complete can be found in \cite{fc-theorem}.\\

\textit{Holder and Sensor.} A Zero-order-Holder (ZoH) and a (ideal) sensor are placed before and after the plant $P$, respectively. We assume that: 

\begin{itemize}
\item[(A.1)] The ZoH and the sensor are synchronized and update their output values at times that are integer multiples of the same interval $\tau\in\mathbb{R}^{+}$, i.e.
\[
\begin{array}
{rl}
u(s \tau+t) =u(s\tau), & y(s \tau+t)  =y(s\tau)  =x(s\tau),
\end{array}
\]
%
for $t\in [0,\tau[$ and $s\in \mathbb{N}_0$, where $s$ is the index of the sampling interval (starting from $0$).
\end{itemize}

\textit{Symbolic controller.} A symbolic controller is a function: 
\[
C:[X]_{\mu_{x}} \rightarrow[U]_{\mu_{u}},
\]
with $\mu_{x},\mu_{u}\in\mathbb{R}^{+}$. In the sequel we suppose that $\mu_{x} \leq \hat{\mu}_X$ and $\mu_{u} \leq \hat{\mu}_U$ so that the domain and co--domain of $C$ are non--empty. If $X$ is bounded, the quantization on $X$ implies that the amount of information associated with any function $C$ so defined is finite. We assume that:
\begin{itemize}
\item[(A.2)] There is a time-varying computation time 
\[
\Delta^{\ctrl}_k\in [\Delta^{\ctrl}_{\min},\Delta^{\ctrl}_{\max}],\,\,\, k\in\mathbb{N}, 
\]
for the symbolic controller to return its output value.
\end{itemize}

\textit{Limited bandwidth.} Let $B_{\max}\in\mathbb{R}^{+}$ be the maximum capacity of the digital communication channel (expressed in bits per second (bps)). Such a constraint imposes a minimum positive `time-to-send', in order to send finite-length information through the communication channel. This requires, in turn, state and input to be quantized before being sent through the network. The minimum sending intervals in the two branches of the network on the feedback loop are given by:
\[
\begin{array}
{cc}
\Delta_{\send}^{\scc} = \frac{\lceil\log_{2} \vert [X]_{\mu_{x}} \vert \rceil}{B_{\max}}, &
\Delta_{\send}^{\ca} = \frac{\lceil\log_{2} \vert [U]_{\mu_{u}} \vert \rceil}{B_{\max}},
\end{array}
\]
where '$\scc$' refers to the sensor-to-controller branch and '$\ca$' to the controller-to-actuator branch of the network.\\

\textit{Time-varying unknown bounded delays.} The actual time occurring for the data to cross the network is larger than the minimum sending time given by the bandwidth requirements. We define the sequence $\{\bar{\Delta}_{k}^{\delay}\}_{k\in\mathbb{N}_0}$ that takes into account time-varying network delays including e.g. congestion, other accesses to the communication channel and any kind of scheduling protocol. The delays induced by the two branches of the network on the feedback loop are:
\[
\begin{array}
{cc}
\Delta_{2k} = \Delta_{\send}^{\scc} + \bar{\Delta}_{2k}^{\delay}, &
\Delta_{2k+1} = \Delta_{\send}^{\ca} + \bar{\Delta}_{2k+1}^{\delay}.
\end{array}
\]
Furthermore we consider a sequence $\{\Delta_{k}^{\req}\}_{k\in\mathbb{N}_0}$ of network \emph{waiting times} that model the delay between the network request and the network access. In the proposed NCS, any scheduling protocol can be considered, provided that it satisfies: 
\begin{itemize}
\item[(A.3)] The sequence of network communication delays is bounded, i.e. 
\[
\bar{\Delta}_{k}^{\delay}\in [\Delta^{\delay}_{\min}, \bar{\Delta}_{\max}^{\delay}],
\]
for all $k\in\mathbb{N}_0$.
\item[(A.4)] The sequence of network waiting times is bounded, i.e.
\[
\Delta_{k}^{\req}\in [0, \Delta_{\max}^{\req}],
\]
for all $k\in\mathbb{N}_0$.
\end{itemize}

\textit{Packet dropout.} Assume that one or more messages can be lost during the transmission through the network. Because of the bounded delays introduced by the network (see Assumptions (A.2), (A.3), (A.4)), if a node does not receive new information within a time less than
\[
\Delta_{\send}^{\scc}+\Delta_{\max}^{\ctrl}+\Delta_{\send}^{\ca}+2\Delta^{\req}_{\max}+2\bar{\Delta}_{\max}^{\delay},
\]
a message is lost. By following the \emph{emulation approach}, see e.g. \cite{HeemelsSurvey}, in dealing with dropout we assume that: 
\begin{itemize}
\item[(A.5)] The maximum number of subsequent dropouts over the network is bounded.
\end{itemize}
The previous assumption allows us to manage packet loss by considering an increased \emph{equivalent delay} $\Delta_{\max}^{\delay}$ introduced by the network, instead of the original $\bar{\Delta}_{\max}^{\delay}$. \\

We now describe recursively the evolution of the NCS, starting from the initial time $t=0$. Consider the $k$--th iteration in the feedback loop. 
The sensor requests access to the network and after a waiting time $\Delta_{2k}^{\req}$, it sends at time $t_{2k}$ the latest available sample $y_k=[y(t_{2k})]_{\mu_{x}}$ where $\mu_{x}$ is the precision of the quantizer that follows the sensor in the NCS scheme (see Figure \ref{NCSpic}). \\
The \emph{sensor-to-controller (sc)} link of the network introduces a delay $\Delta_{2k}$, after which the sample reaches the controller that computes in $\Delta^{\ctrl}_{k}$ time units the value $u_{k+1}=C(y_k)$. The controller requests access to the network and sends the control sample $u_{k+1}$ at time $t_{2k+1}$ (after a bounded waiting time $\Delta_{2k+1}^{\req}$). \\
The \emph{controller-to-actuator (ca)} link of the network introduces a delay $\Delta_{2k+1}$, after which the sample reaches the ZoH. At time $t=A_{k+1} \tau$ the ZoH is refreshed to the control value $u_{k+1}$ where $ A_{k+1}:=\lceil (t_{2k+1}+\Delta_{k}^{\ca})/\tau\rceil$. The next iteration starts and the sensor requests access to the network again.

Consider now the sequence of control values $\{u_{k}\}_{k\in\mathbb{N}_0}$. Each value is held up for $N_k:=A_{k+1}-A_{k}$ sampling intervals. Due to the bounded delays, one gets:
\[
N_k \in [{N}_{\min}; {N}_{\max}],
\]
with:
\begin{align}
{N}_{\min} &=\left\lceil \Delta_{\min}/\tau \right\rceil, \label{eq:N_min}\\
{N}_{\max}  &= \left\lceil \Delta_{\max}/\tau \right \rceil,\label{eq:N_max}
\end{align}
where we set: 
\begin{align}
\Delta_{\min} & :=\Delta_{\send}^{\scc}+\Delta_{\min}^{\ctrl}+\Delta_{\send}^{\ca}+2\Delta_{\min}^{\delay},\label{eq:delta_min}\\
\Delta_{\max} & :=\Delta_{\send}^{\scc}+\Delta_{\max}^{\ctrl}+\Delta_{\send}^{\ca}+2\Delta^{\req}_{\max}+2\Delta_{\max}^{\delay}\label{eq:delta_max}.
\end{align}
For later purposes we collect the computation and communication parameters appearing in the previous description in the following vector:
\begin{equation}\label{eq:NCS_parameters}
\mathcal{C}_{\text{NCS}}=(\tau,\mu_{x},\mu_{u},B_{\max},\Delta_{\min},\Delta_{\max}).
\end{equation}
In the sequel we refer to the described NCS by $\Sigma$. The collection of trajectories of the plant $P$ in the NCS $\Sigma$ is denoted by $\text{Traj}(\Sigma)$. Moreover we refer to a trajectory of $\Sigma$ with initial state $x_0$ and control input $u$ by $\mathbf{x}(.,x_{0},u)$.

\section{Systems, approximate equivalence and composition}\label{subsec:ApproxEquiv}
We will use the notion of systems as a unified mathematical framework to describe networked control systems as well as their symbolic models.

\begin{definition}
\cite{paulo} 
A system $S$ is a sextuple: 
\begin{equation}\label{def:system}
S=(X,X_{0},U,\rTo,Y,H),
\end{equation}
consisting of: 

\begin{itemize}
\item a set of states $X$;
\item a set of initial states $X_{0}\subseteq X$;
\item a set of inputs $U$;
\item a transition relation $\rTo \subseteq X\times U\times X$;
\item a set of outputs $Y$;
\item an output function $H:X\rightarrow Y$. 
\end{itemize}
A transition $(x,u,x^{\prime})\in\rTo$ is denoted by $x\rTo^{u}x^{\prime}$. For such a transition, state $x^{\prime}$ is called a $u$-successor, or simply a successor, of state $x$. The set of $u$-successors of a state $x$ is denoted by $\Post_u (x)$.
\end{definition}

A state run of $S$ is a (possibly infinite) sequence of transitions:
\begin{equation}
x_{0} \rTo^{u_{1}} x_{1} \rTo^{u_{2}} \,\, {...}
\label{eq:state_run}
\end{equation}
with $x_0\in X_0$. An output run is a (possibly infinite) sequence $\{y_i\}_{i\in\mathbb{N}_0}$ such that there exists a state run of the form (\ref{eq:state_run}) with $y_i=H(x_i)$, $i\in\mathbb{N}_0$. 
System $S$ is said to be:
\begin{itemize}
\item \textit{countable}, if $X$ and $U$ are countable sets;
\item \textit{symbolic}, if $X$ and $U$ are finite sets;
\item \textit{metric}, if the output set $Y$ is equipped with a metric $d:Y\times Y\rightarrow\mathbb{R}_{0}^{+}$;
\item \textit{deterministic}, if for any $x\in X$ and $u\in U$ there exists at most one state $x^{\prime}\in X$ such that $x \rTo^{u} x^{\prime}$;
\item \textit{non--blocking}, if for any $x\in X$ there exists at least one state $x^{\prime}\in X$ such that $x \rTo^{u} x^{\prime}$ for some $u\in U$.
\end{itemize}
\begin{definition}
Given two systems $S_{i}=(X_{i},X_{0,i},U_{i},$ $\rTo_{i},Y_{i},H_{i})$ ($i=1,2$), $S_{1}$ is a \textit{sub--system} of $S_{2}$, denoted $S_{1} \sqsubseteq S_{2}$, if $X_{1}\subseteq X_{2}$, $X_{0,1}\subseteq X_{0,2}$, $U_{1}\subseteq U_{2}$, $\rTo_{1}\subseteq \rTo_{2}$, $Y_{1}\subseteq Y_{2}$, $H_{1}(x)=H_{2}(x)$ for any $x\in X_{1}$.
\end{definition}

In the sequel we consider bisimulation relations \cite{Milner,Park} to relate properties of networked control systems and symbolic models. Intuitively, a bisimulation relation between a pair of systems $S_{1}$ and $S_{2}$ is a relation between the corresponding state sets explaining how a state run $r_{1}$ of $S_{1}$ can be transformed into a state run $r_{2}$ of $S_{2}$ and vice versa. While typical bisimulation relations require that $r_{1}$ and $r_{2}$ share the same output run, the notion of approximate bisimulation, introduced in \cite{AB-TAC07}, relaxes this condition by requiring the outputs of $r_{1}$ and $r_{2}$ to simply be close, where closeness is measured with respect to the metric on the output set. 

\begin{definition}
\cite{AB-TAC07}
\label{ASR} 
Let \mbox{$S_{i}=(X_{i},X_{0,i},U_{i},\rTo_{i},Y_{i},H_{i})$} ($i=1,2$) be metric systems with the same output sets $Y_{1}=Y_{2}$ and metric $d$ and consider a precision $\varepsilon\in\mathbb{R}^{+}_{0}$. A relation $\mathcal{R}\subseteq X_{1}\times X_{2}$ is an $\varepsilon$--approximate simulation relation from $S_{1}$ to $S_{2}$ if the following conditions are satisfied:

\begin{itemize}
\item[(i)] for every $x_{1}\in X_{0,1}$ there exists $x_{2}\in X_{0,2}$ such that 
$(x_{1},x_{2})\in \mathcal{R}$;
\item[(ii)] for every $(x_{1},x_{2})\in \mathcal{R}$ we have \mbox{$d(H_{1}(x_{1}),H_{2}(x_{2}))\leq\varepsilon$};
\item[(iii)] for every $(x_{1},x_{2})\in \mathcal{R}$ the existence of \mbox{$x_{1}\rTo_{1}^{u_{1}}x'_{1}$ in $S_{1}$} implies the existence of \mbox{$x_{2}\rTo_{2}^{u_{2}}x'_{2}$} in $S_{2}$ satisfying $(x^{\prime}%
_{1},x^{\prime}_{2})\in \mathcal{R}$.
\end{itemize}

System $S_{1}$ is $\varepsilon$--simulated by $S_{2}$ or $S_{2}$ $\varepsilon$--simulates $S_{1}$, denoted \mbox{$S_{1}\preceq_{\varepsilon}S_{2}$}, if there exists an $\varepsilon$--approximate simulation relation from $S_{1}$ to $S_{2}$. The relation $\mathcal{R}$ is an \mbox{$\varepsilon$--approximate} bisimulation relation between $S_{1}$ and $S_{2}$ if $\mathcal{R}$ is an \mbox{$\varepsilon$--approximate} simulation relation from $S_{1}$ to $S_{2}$ and $\mathcal{R}^{-1}$ is an \mbox{$\varepsilon$--approximate} simulation relation from $S_{2}$ to $S_{1}$. Furthermore, systems $S_{1}$ and $S_{2}$ are $\varepsilon$--bisimilar, denoted \mbox{$S_{1}\cong_{\varepsilon}S_{2}$}, if there exists an $\varepsilon$--approximate bisimulation relation $\mathcal{R}$ between $S_{1}$ and $S_{2}$. When $\varepsilon=0$ systems $S_{1}$ and $S_{2}$ are said to be exactly bisimilar.
\end{definition}

In this work we also consider a generalization of approximate bisimulation, called alternating approximate bisimulation, that has been introduced in \cite{PolaSIAM2009} to relate properties of control systems affected by non-determinism and their symbolic models. 

\begin{definition}
\label{ASR_S} 
\cite{PolaSIAM2009,paulo}
Let $S_{i}=(X_{i},X_{0,i},U_{i},\rTo_{i},$ $Y_{i},H_{i})$ ($i=1,2$) be metric systems with the same output sets $Y_{1}=Y_{2}$ and metric $d$ and consider a precision $\varepsilon\in\mathbb{R}^{+}_{0}$. A relation $\mathcal{R}\subseteq X_{1}\times X_{2}$ is an alternating $\varepsilon$--approximate ($A\varepsilon A$) simulation relation from $S_{1}$ to $S_{2}$ if the following conditions are satisfied:
\begin{itemize}
\item[(i)] for every $x_{1}\in X_{0,1}$ there exists $x_{2}\in X_{0,2}$ such that
$(x_{1},x_{2})\in \mathcal{R}$;
\item[(ii)]  for every $(x_{1},x_{2})\in\mathcal{R}$ we have $d(H_{1}(x_{1}),H_{2}(x_{2}))\leq\varepsilon$;
\item[(iii)] for every $(x_{1},x_{2})\in\mathcal{R}$ and for every $u_{1}\in U_{1}$ there exists $u_{2}\in U_{2}$ such that for every $x_{2}^{\prime}\in \Post_{u_2}(x_2)$ there exists $x_{1}^{\prime}\in \Post_{u_1}(x_1)$ satisfying $(x_{1}^{\prime},x_{2}^{\prime} )\in\mathcal{R}$.
\end{itemize}
System $S_{1}$ is alternating $\varepsilon$--simulated by $S_{2}$ or $S_{2}$ alternating $\varepsilon$--simulates $S_{1}$, denoted \mbox{$S_{1}\preceq_{\varepsilon}^{\alt} S_{2}$}, if there exists an $A\varepsilon A$ simulation relation from $S_{1}$ to $S_{2}$. 
Relation $\mathcal{R}$ is an $A \varepsilon A$ bisimulation relation between $S_{1}$ and $S_{2}$ if $\mathcal{R}$ is an $A \varepsilon A$ simulation relation from $S_{1}$ to $S_{2}$ and $\mathcal{R}^{-1}$ is an $A \varepsilon A$ simulation relation from $S_{2}$ to $S_{1}$. Furthermore, systems $S_{1}$ and $S_{2}$ are $A \varepsilon A$--bisimilar, denoted \mbox{$S_{1}\cong_{\varepsilon}^{\alt} S_{2}$}, if there exists an $A \varepsilon A$ bisimulation relation $\mathcal{R}$ between $S_{1}$ and $S_{2}$.
\end{definition}

When $\varepsilon=0$, the above notion can be viewed as the two-player version of the notion of alternating bisimulation \cite{alternating}. 
We conclude this section by introducing the notion of approximate parallel composition proposed in \cite{TabuadaTAC08} that is employed in the sequel to capture (feedback) interaction between systems and symbolic controllers. 

\begin{definition}
\cite{TabuadaTAC08}
\label{composition} 
Consider a pair of metric systems \mbox{$S_{i}=(X_{i},X_{0,i},U_{i},\rTo_{i},Y_{i},H_{i})$} ($i=1,2$) with the same output sets $Y_{1}=Y_{2}$ and metric $d$, and a parameter $\theta\in\mathbb{R}_{0}^{+}$. The $\theta$--approximate parallel composition of $S_{1}$ and $S_{2}$ is the system
\[
S_{1}\Vert_{\theta}S_{2}=(X,X_{0},U,\rTo,Y,H), 
\]
where:
\begin{itemize}
\item $X=\{(x_{1},x_{2})\in X_{1}\times X_{2}\text{ } |\text{ } d(H_{1}(x_{1}),H_{2}(x_{2}))\leq \theta\}$;
\item $X_{0}=X\cap(X_{0,1}\times X_{0,2})$;
\item $U=U_{1}\times U_{2}$;
\item $(x_{1},x_{2})\rTo^{(u_{1},u_{2})}(x_{1}^{\prime},x_{2}^{\prime})$ if $x_{1}\rTo_{1}^{u_{1}}x_{1}^{\prime}$ and $x_{2}\rTo_{2}^{u_{2}}x_{2}^{\prime}$;
\item $Y=Y_{1}$;
\item $H(x_{1},x_{2})=H_{1}(x_{1})$ for any $(x_{1},x_{2})\in X$.
\end{itemize}
\end{definition}

The interested reader is referred to \cite{TabuadaTAC08,paulo} for a detailed description of the notion of approximate parallel composition and of its properties.

\section{Symbolic models for NCS}\label{sec:SymbolicModels}

In this section we propose symbolic models that approximate NCS in the sense of alternating approximate bisimulation. 
For notational simplicity we denote by $u$ any constant control input $\tilde{u}\in\mathcal{U}$ s.t. $\tilde{u}(t)=u$ for all times $t\in\mathbb{R}^{+}$. Set 
\[
X_e=\bigcup_{N\in [N_{\min};N_{\max}]} X^{N}.
\]
Given the NCS $\Sigma$ and the vector $\mathcal{C}_{\text{NCS}}$ of parameters in (\ref{eq:NCS_parameters}), consider the following system:

\begin{equation*}
S(\Sigma):=(X_{\tau},X_{0,\tau},U_{\tau},\rTo_{\tau},Y_{\tau},H_{\tau}),
\label{systemTD}
\end{equation*}
where:

\begin{itemize}
\item $X_{\tau}$ is the subset of $ X_0 \cup X_e$ such that for any $x=(x_{1},x_{2},...,x_{N})\in X_{\tau}$, with $N\in [N_{\min};N_{\max}]$, the following conditions hold:
\begin{align}
x_{i+1} &=\mathbf{x}(\tau,x_{i},u^{-}), \qquad i\in [1;N-2];\label{eq:states_1}\\
x_{N} & =\mathbf{x}(\tau,x_{N-1},u^{+})\label{eq:states_2};
\end{align}
for some constant functions $u^{-}$, $u^{+}\in [U]_{\mu_{u}}$.
\item $X_{0,\tau}=X_0$;
\item $U_{\tau}=[U]_{\mu_{u}}$;
\item $x^{1}\rTo^{u}_{\tau} x^{2}$, where:
\[
\begin{array}
{l}
\left\{
\begin{array}
{llll}
x_{i+1}^{1} & = & \mathbf{x}(\tau,x_{i}^{1},u^{-}_{1}), & i\in [1;N_{1}-2];\\
x_{N_{1}}^{1} & = & \mathbf{x}(\tau,x^1_{N_{1}-1},u^{+}_{1}); &
\end{array}
\right.
\\
\\
\left\{
\begin{array}
{llll}
x_{i+1}^{2} & = & \mathbf{x}(\tau,x_{i}^{2},u^{-}_{2}), & i\in [1;N_{2}-2];\\
x_{N_{2}}^{2} & = & \mathbf{x}(\tau,x^2_{N_{2}-1},u_2^{+}); &
\end{array}
\right.
\\
\\
\left\{
\begin{array}
{llll}
u_2^{-} & = & u_1^{+}; &\\
u_2^{+} & = & u; &\\

x^{2}_{1} & = & \mathbf{x}(\tau,x^1_{N_{1}},u^{-}_{2}); &
\end{array}
\right.
\end{array}
\]
for some $N_{1},N_{2}\in [N_{\min};N_{\max}]$;
\item $Y_{\tau}=X_{\tau}$;
\item $H_{\tau}=1_{X_{\tau}}$.
\end{itemize}

Note that $S(\Sigma)$ is non-deterministic because, depending on the values of $N_{2}$, more than one $u$--successor of $x^{1}$ may exist. The construction of the set of states of $S(\Sigma)$ is based on an extended-state-space approach, and has been inspired by known approaches in the analysis of discrete--time time--varying delay systems, see e.g. \cite{TDSDT}.  
Since the state vectors of $S(\Sigma)$ are built from trajectories of $\Sigma$ sampled every $\tau$ time units, $S(\Sigma)$ collects all the information of the NCS $\Sigma$ available at the sensor (see Figure \ref{NCSpic}) as formally stated in the following result.
\begin{theorem}
Given the NCS $\Sigma$ and the system $S(\Sigma)$ the following properties hold:
\begin{itemize}
\item for any trajectory $\mathbf{x}(.,x_{0},u)\in\text{Traj}(\Sigma)$ of $\Sigma$, there exists a state run 
\begin{equation}
x^{0} \rTo^{u_{1}} x^{1} \rTo^{u_{2}} \,\, {...}\,, 
\label{cond1}
\end{equation}
of $S(\Sigma)$ with $x^{i}=(x_{1}^{i},x_{2}^{i},...,x_{N_{i}}^{i})$ such that $x^{0}=x_{0}$ and the sequence of states 
\begin{equation}
\begin{array}
{rclclcl}
{x^{0}} & , & \underbrace{x^{1}_{1}, {...}, x^{1}_{N_{0}+1}}_{x^{1}} & , & \underbrace{x^{2}_{1}, {...}, x^{2}_{N_{1}}}_{x^{2}} & , & {...}
\end{array}
\label{cond2}
\end{equation}

obtained by concatenating each component of the vectors $x^{i}$, coincides with the sequence of sensor measurements
\begin{align}
& y(0), y(\tau), {...}, y((N_{0}+1)\tau),y((N_{0}+2)\tau), {...}, \nonumber\\
& y((N_{0}+N_{1}+1)\tau), {...}
\label{cond3}
\end{align}
in the NCS $\Sigma$;
\item for any state run (\ref{cond1}) of $S(\Sigma)$, there exists a trajectory $\mathbf{x}(.,x_{0},u)\in\text{Traj}(\Sigma)$ of $\Sigma$ such that the sequence of states in (\ref{cond2}) coincides with the sequence (\ref{cond3}) of sensor measurements in the NCS $\Sigma$.
\end{itemize}
\end{theorem}

The proof of the above result is a direct consequence of the definition of $S(\Sigma)$ and is therefore omitted. 
System $S(\Sigma)$ can be regarded as a metric system with the metric $d_{Y_{\tau}}$ on $Y_{\tau}$ naturally induced by the metric $d_{X}(x_{1},x_{2})=\Vert x_{1} - x_{2}\Vert$ on $X$, as follows. Given any $x^{i}=(x_{1}^{i},x_{2}^{i},...,x_{N_{i}}^{i})$, $i=1,2$, we set:
\begin{equation}\label{eq:metric}
d_{Y_{\tau}}(x^{1},x^{2}):=
\left\{
\begin{array}
{ll}
\max_{i\in[1;N]}\Vert x^{1}_{i}-x^{2}_{i}\Vert, & \text{ if } N_{1}= N_{2}=N. \\
+\infty, & \text{ otherwise.}
\end{array}
\right.
\end{equation}

Although system $S(\Sigma)$ contains all the information of the NCS $\Sigma$ available at the sensor, it is not a finite model. We now propose a system which approximates $S(\Sigma)$ and is symbolic. 
Define the following system:

\begin{equation}
S_{\ast}(\Sigma):=(X_{\ast},U_{\ast},\rTo_{\ast},Y_{\ast},H_{\ast}), 
\label{symbmodel}
\end{equation}
where:

\begin{itemize}
\item $X_{\ast}$ is the subset of $[X_0 \cup X_e]_{\mu_{x}}$ such that for any $x^{\ast}=(x^{\ast}_{1},x^{\ast}_{2},...,x^{\ast}_{N})\in X_{\ast}$, with $N\in [N_{\min};N_{\max}]$, the following condition holds:
\begin{align}
x^{\ast}_{i+1} &=[\mathbf{x}(\tau,x^{\ast}_{i},u^{-}_{\ast})]_{\mu_{x}}, \qquad i\in [1;N-2];\label{eq:symb_states_1}\\
x^{\ast}_{N} & =[\mathbf{x}(\tau,x^{\ast}_{N-1},u^{+}_{\ast})]_{\mu_{x}};\label{eq:symb_states_2}
\end{align}
for some constant functions $u^{-}_{\ast}$, $u^{+}_{\ast}\in [U]_{\mu_{u}}$.

\item $X_{0,\ast}=[X_0]_{\mu_{x}}$;

\item $U_{\ast}=[U]_{\mu_{u}}$;

\item $x^{1}\rTo^{u_{\ast}}_{\ast} x^{2}$, where:
\[
\begin{array}
{l}
\left\{
\begin{array}
{llll}
x_{i+1}^{1} & = & [\mathbf{x}(\tau,x_{i}^{1},u^{-}_{1})]_{\mu_{x}}, & i\in [1;N_{1}-2];\\
x_{N_{1}}^{1} & = & [\mathbf{x}(\tau,x^1_{N_{1}-1},u^{+}_{1})]_{\mu_{x}}; &
\end{array}
\right.
\\
\\
\left\{
\begin{array}
{llll}
x_{i+1}^{2} & = & [\mathbf{x}(\tau,x_{i}^{2},u^{-}_{2})]_{\mu_{x}}, & i\in [1;N_{2}-2];\\
x_{N_{2}}^{2} & = & [\mathbf{x}(\tau,x^2_{N_{2}-1},u_2^{+})]_{\mu_{x}}; &
\end{array}
\right.
\\
\\
\left\{
\begin{array}
{llll}
u_2^{-} & = & u_1^{+}; &\\
u_2^{+} & = & u_{\ast}; &\\
x^{2}_{1} & = & [\mathbf{x}(\tau,x^1_{N_{1}},u^{-}_{2})]_{\mu_{x}}; &
\end{array}
\right.
\end{array}
\]
for some $N_{1},N_{2}\in [N_{\min};N_{\max}]$;
 
\item $Y_{\ast}=X_{\tau}$;
\item $H_{\ast}=\imath:X^{\ast}\hookrightarrow Y_{\ast}$.
\end{itemize}

System $S_{\ast}(\Sigma)$ is metric when we regard the set of outputs $Y_{\ast}$ as being equipped with the metric in (\ref{eq:metric}).

\begin{remark}
System $S_{\ast}(\Sigma)$ is countable and becomes symbolic when the set of states $X$ is bounded. 
This model can be constructed in a finite number of steps, as inferable from its definition. Space complexity in storing data of $S_{\ast}(\Sigma)$ is generally rather large, because of the large size of the set of states $X_{e}$. This choice in the definition of $X_e$ makes it easier to compare the NCS and $S_{\ast}(\Sigma)$ in terms of alternating approximate bisimulation as we will see in the forthcoming developments (see Theorem \ref{thmain}). However, for computational purposes it is possible to give a more concise representation of $X_e$ as follows: any state $(x_1,x_2,\dots ,x_N)$ in $X_e$ can be equivalently represented by the tuple $(x_1,u^{-},u^{+},N)$ where $u^{-}$ and $u^{+}$ are the control inputs in Eqns. (\ref{eq:states_1})--(\ref{eq:states_2}). 
\end{remark}
\begin{remark}
While the semantics of the NCS $\Sigma$ is described in closed--loop, the symbolic models in (\ref{symbmodel}) approximate the NCS in open--loop. Indeed, the symbolic models proposed approximate the plant $P$ and the communication network, i.e. all entities in the NCS feedback loop except for the symbolic controller $C$ (see Figure \ref{NCSpic}). This choice allows us to view the closed--loop NCS as the parallel composition \cite{ModelChecking} of two symbolic systems and therefore to adapt standard results in computer science for the control design of NCS, as shown in Section \ref{sec:control}. 
\end{remark}
A key ingredient of our results is the notion of incremental global
asymptotic stability that we report hereafter. 

\begin{definition}
\label{dGAS}
\cite{IncrementalS} 
Control system (\ref{NCSeq1}) is incrementally globally asymptotically stable ($\delta$--GAS) if it is forward complete and there exist a $\mathcal{KL}$ function $\beta$ and a $\mathcal{K}_{\infty}$ function $\gamma$ such that for any $t\in{\mathbb{R}_{0}^{+}}$, any $x_{1},x_{2}\in X$ and any
$u\in\mathcal{U}$, the following condition is satisfied:
\begin{equation*}
\Vert \mathbf{x}(t,x_1,u)-\mathbf{x}(t,x_2,u) \Vert\leq \beta(\Vert x_{1}-x_{2}\Vert ,t).
\label{delta_GAS}
\end{equation*}
\end{definition}

The above incremental stability notion can be characterized in terms of dissipation inequalities, as follows.
\begin{definition}
\label{dGAS_Lyapunov}
\cite{IncrementalS} 
A smooth function $V:X\times X\rightarrow\mathbb{R}$ is called a $\delta$--GAS Lyapunov function for the control system (\ref{NCSeq1})  if there exist $\lambda\in\mathbb{R}^{+}$ and $\mathcal{K}_{\infty}$ functions $\underline{\alpha}$ and $\overline{\alpha}$ such that, for any $x_{1},x_{2}\in X$ and any $u\in U$, the following conditions hold true:
\begin{itemize}
\item[(i)] $\underline{\alpha}(\Vert{x_{1}-x_{2}}\Vert)\leq V(x_{1},x_{2})\leq\overline{\alpha}(\Vert{x_{1}-x_{2}}\Vert)$,
\item[(ii)] $\frac{\partial{V}}{\partial{x_{1}}} f(x_{1},u)+\frac{\partial{V}}{\partial{x_{2}}} f(x_{2},u) \leq -\lambda V(x_{1},x_{2})$.
\end{itemize}
\end{definition}

The following result adapted from \cite{IncrementalS} completely characterizes $\delta$--GAS in terms of existence
of $\delta$--GAS Lyapunov functions.

\begin{theorem}
\label{TH-IGAS}
Control system (\ref{NCSeq1}) is $\delta$--GAS if and only if it admits a $\delta$--GAS Lyapunov function. 
\end{theorem}

\begin{remark}
In this paper we assume that the nonlinear control system $P$ is $\delta$--GAS. Backstepping techniques for the incremental stabilization of nonlinear control systems have been recently proposed in \cite{MajidBackstepping}.
\end{remark}

We now have all the ingredients to present the main result of this section.

\begin{theorem}
\label{thmain}
Consider the NCS $\Sigma$ and suppose that the control system $P$ enjoys the following properties:
\begin{itemize}
\item[(H1)] There exists a $\delta$--GAS Lyapunov function satisfying the inequality (ii) in Definition \ref{dGAS_Lyapunov} for some $\lambda\in\mathbb{R}^{+}$;
\item[(H2)] There exists a $\mathcal{K}_{\infty}$ function $\gamma$ such that\footnote{Note that since $V$ is smooth, if the state space $X$ is bounded, which is the case in many concrete applications, one can always choose $\gamma(\Vert w-z \Vert)= \left( \sup_{x,y\in X} \Vert \frac{\partial{V}}{\partial{y}}(x,y) \Vert \right) \Vert w-z \Vert$.}:
\[
V(x,x^{\prime})-V(x,x^{\prime \prime})\leq\gamma(\Vert{x^{\prime}-x^{\prime \prime}}\Vert),
\]
for every $x,x^{\prime},x^{\prime\prime}\in X$.
\end{itemize}
For any desired precision $\varepsilon\in\mathbb{R}^{+}$, sampling time $\tau\in\mathbb{R}^{+}$ and state quantization $\mu_{x}\in\mathbb{R}^{+}$ satisfying the following inequality:
\begin{equation}
\label{bisim_condition1}
\mu_{x} \leq \min \left \{\gamma^{-1} \left ( \left ( 1-e^{-\lambda \tau} \right ) \underline{\alpha}(\varepsilon) \right ) , \overline{\alpha}^{-1}(\underline{\alpha}(\varepsilon)), \hat{\mu}_{X} \right \}\text{,}
\end{equation}

systems $S(\Sigma)$ and $S_{\ast}(\Sigma)$ are $A\varepsilon A$--bisimilar. 
\label{polaut}
\end{theorem}


\begin{proof}
Consider the relation $\mathcal{R}\subseteq X_{\tau}\times X_{\ast}$ defined by $(x,x^{\ast})\in\mathcal{R}$ if and only if:
\begin{itemize}
\item $x=(x_{1},x_{2},...,x_{N})$, $x^{\ast}=(x^{\ast}_{1},x^{\ast}_{2},...,x^{\ast}_{N})$, for some $N\in [N_{\min};N_{\max}]$;
\item $V(x_{i},x_{i}^{\ast})\leq\underline{\alpha}(\varepsilon)$ for $i\in [1;N]$;
\item Eqns. (\ref{eq:states_1}), (\ref{eq:states_2}), (\ref{eq:symb_states_1}), (\ref{eq:symb_states_2}) hold for some $u^{-}=u^{-}_{\ast}$ and $u^{+}=u^{+}_{\ast}$.
\end{itemize}
In the following we prove that $S(\Sigma) \preceq_{\varepsilon}^{\alt} S_{\ast}(\Sigma)$, according to Definition \ref{ASR_S}. 
We first prove condition (i) of Definition \ref{ASR_S}. For any $x\in X_{0,\tau}$, choose $x^{\ast}\in X_{0,\ast}$ such that $x^{\ast}=[x]_{\mu_{x}}$, which implies that $\Vert x^{\ast}-x \Vert \leq \mu_{x}$. Hence, from condition (i) in Definition \ref{dGAS_Lyapunov} and the inequality in (\ref{bisim_condition1}) one gets:
\begin{equation}
V(x,x^{\ast})\leq \overline{\alpha}(\mu_{x}) \leq \overline{\alpha}(\overline{\alpha}^{-1}(\underline{\alpha}(\varepsilon))) = \underline{\alpha}(\varepsilon),
\label{eq:bound_distance}
\end{equation}
which concludes the proof of condition (i). We now consider condition (ii) of Definition \ref{ASR_S}. For any $(x,x^{\ast})\in\mathcal{R}$, from the definition of the metric given in (\ref{eq:metric}), the definition of $\mathcal{R}$ and condition (i) in Definition \ref{dGAS_Lyapunov}, one can write:
\begin{align}
d_{Y_{\tau}}(x,x^{\ast}) & =\max_{i}\Vert x_{i}-x^{\ast}_{i}\Vert \leq \max_{i} \underline{\alpha}^{-1}(V(x_{i},x_{i}^{\ast})) \nonumber\\
& \leq \underline{\alpha}^{-1}(\underline{\alpha}(\varepsilon)) = \varepsilon. \nonumber
\end{align}
Next we show that condition (iii) in Definition \ref{ASR_S} holds. Consider any $(x,x^{\ast})\in\mathcal{R}$, with $x=(x_{1},x_{2},...,x_{N})$, $x^{\ast}=(x^{\ast}_{1},x^{\ast}_{2},...,x^{\ast}_{N})$, for some $N\in [N_{\min};N_{\max}]$, and any $u\in U_{\tau}$; then pick $u_{\ast}=u\in U_{\ast}$. Now consider any $\bar{x}^{\ast}=(\bar{x}^{\ast}_1,\bar{x}^{\ast}_2,...,\bar{x}^{\ast}_{\bar{N}})\in \Post_{u_{\ast}}(x^{\ast}) \subseteq X_{\ast}$ with $\bar{x}^{\ast}_{\bar{N}}=[\mathbf{x}(\tau,\bar{x}^{\ast}_{\bar{N}-1},u_{\ast})]_{\mu_{x}}$, for some $\bar{N}\in [N_{\min};N_{\max}]$. Pick $\bar{x}=(\bar{x}_1,\bar{x}_2,...,\bar{x}_{\bar{N}})\in \Post_{u}(x) \subseteq X_{\tau}$ with $\bar{x}_{\bar{N}}=\mathbf{x}(\tau,\bar{x}_{\bar{N}-1},u)$ and define the state $\tilde{x}^{\ast}_1:=\mathbf{x}(\tau,x^{\ast}_{N},u_{\ast}^+)$. 
By Assumption (H1), condition (ii) in Definition \ref{dGAS_Lyapunov} writes:
\begin{equation}
\frac{\partial{V}}{\partial{x_{N}}} f(x_{N},u^+)+\frac{\partial{V}}{\partial{x_{N}^{\ast}}} f(x_{N}^{\ast},u_{\ast}^+)  \leq -\lambda V(x_{N},x_{N}^{\ast}). \\
\label{eq:ineq}
\end{equation}

By considering Assumption (H2), the definitions of $\mathcal{R}$, $S(\Sigma)$ and $ S_{\ast}(\Sigma)$, and by integrating the previous inequality, the following holds:
\begin{equation}
\label{eq_chain_of_ineq}
\begin{array}
{rcl}
V(\bar{x}_{1},\bar{x}^{\ast}_{1})  & \leq & V(\bar{x}_{1},\tilde{x}_1^{\ast}) +\gamma(\Vert{\tilde{x}_1^{\ast}-\bar{x}^{\ast}_{1}}\Vert) \\
& \leq & e^{-\lambda \tau} V(x_N,x^{\ast}_N)+\gamma(\Vert{\tilde{x}_1^{\ast}-\bar{x}^{\ast}_{1}}\Vert) \\
& \leq & e^{-\lambda \tau} \underline{\alpha}(\varepsilon)+\gamma(\mu_{x}) \leq \underline{\alpha}(\varepsilon),
\end{array}
\end{equation}
where condition (\ref{bisim_condition1}) has been used in the last step. By similar computations, it is possible to prove by induction that $V(\bar{x}_{i},\bar{x}^{\ast}_{i})  \leq \underline{\alpha}(\varepsilon)$ implies $V(\bar{x}_{i+1},\bar{x}^{\ast}_{i+1})  \leq \underline{\alpha}(\varepsilon)$, for any $i\in [1;\bar{N}-2]$. The last step $i=\bar{N}-1$ requires the use of the input $u=u_{\ast}$ instead of $u^+=u_{\ast}^+$. By Assumption (H1) and defining $\tilde{x}^{\ast}_{\bar{N}}  :=\mathbf{x}(\tau,\bar{x}^{\ast}_{\bar{N}-1},u_{\ast})$, condition (ii) in Definition \ref{dGAS_Lyapunov} writes:
\begin{equation}
\frac{\partial{V}}{\partial{\bar{x}_{\bar{N}}}} f(\bar{x}_{\bar{N}},u)+\frac{\partial{V}}{\partial{\bar{x}^{\ast}_{\bar{N}}}} f(\bar{x}^{\ast}_{\bar{N}},u_{\ast})  \leq -\lambda V(\bar{x}_{\bar{N}},\bar{x}^{\ast}_{\bar{N}}). \\
\label{eq:ineq2}
\end{equation}
By considering Assumption (H2), the definitions of $\mathcal{R}$, $S(\Sigma)$ and $ S_{\ast}(\Sigma)$, and by integrating the previous inequality, the following holds:
\begin{equation}
\label{eq_chain_of_ineq_2}
\begin{array}
{rcl}
V(\bar{x}_{\bar{N}},\bar{x}^{\ast}_{\bar{N}}) & \leq & V(\bar{x}_{\bar{N}},\tilde{x}^{\ast}_{\bar{N}}) +\gamma(\Vert{\tilde{x}^{\ast}_{\bar{N}}-\bar{x}^{\ast}_{\bar{N}}}\Vert) \\
& \leq & e^{-\lambda \tau} V(\bar{x}_{\bar{N}-1},\bar{x}^{\ast}_{\bar{N}-1})+\gamma(\Vert{\tilde{x}^{\ast}_{\bar{N}}-\bar{x}^{\ast}_{\bar{N}}}\Vert) \\
& \leq & e^{-\lambda \tau} \underline{\alpha}(\varepsilon)+\gamma(\mu_{x}) \leq \underline{\alpha}(\varepsilon).
\end{array}
\end{equation}
Hence the inequality $V(\bar{x}_{i},\bar{x}^{\ast}_{i})  \leq \underline{\alpha}(\varepsilon)$ has been proven for any $i\in [1;\bar{N}]$, implying $(\bar{x},\bar{x}^{\ast})\in\mathcal{R}$, which concludes the proof of condition (iii) of Definition \ref{ASR_S}.

We now consider the relation $\mathcal{R}^{-1}$ and we complete the prove by showing that $S_{\ast}(\Sigma) \preceq_{\varepsilon}^{\alt} S(\Sigma)$, according to Definition \ref{ASR_S};  we first prove condition (i) of Definition \ref{ASR_S}. For any $x^{\ast}\in X_{0,\ast}$, choose $x=x^{\ast}\in X_{0,\tau}$, which implies that $\Vert x^{\ast}-x \Vert = 0 \leq \mu_{x}$. Hence the inequality in (\ref{eq:bound_distance}) holds, which concludes the proof of condition (i). The proof of condition (ii) of Definition \ref{ASR_S} for the relation $\mathcal{R}^{-1}$ is the same as the one for the relation $\mathcal{R}$ and is not reported. Next we show that condition (iii) in Definition \ref{ASR_S} holds. Consider any $(x^{\ast},x)\in\mathcal{R}^{-1}$, with $x^{\ast}=(x^{\ast}_{1},x^{\ast}_{2},...,x^{\ast}_{N})$, $x=(x_{1},x_{2},...,x_{N})$, for some $N\in [N_{\min};N_{\max}]$, and any $u_{\ast}\in U_{\ast}$; then pick $u=u_{\ast}\in U_{\tau}$. Now consider any $\bar{x}=(\bar{x}_1,\bar{x}_2,...,\bar{x}_{\bar{N}})\in \Post_{u}(x) \subseteq X_{\tau}$ with $\bar{x}_{\bar{N}}=\mathbf{x}(\tau,\bar{x}_{\bar{N}-1},u)$, for some $\bar{N}\in [N_{\min};N_{\max}]$. Pick $\bar{x}^{\ast}=(\bar{x}^{\ast}_1,\bar{x}^{\ast}_2,...,\bar{x}^{\ast}_{\bar{N}})\in \Post_{u_{\ast}}(x^{\ast}) \subseteq X_{\ast}$ with $\bar{x}^{\ast}_{\bar{N}}=[\mathbf{x}(\tau,\bar{x}_{\bar{N}-1}^{\ast},u_{\ast})]_{\mu_x}$ and define the state $\tilde{x}^{\ast}_1:=\mathbf{x}(\tau,x^{\ast}_{N},u_{\ast}^+)$. After that, it is possible to rewrite exactly the same steps as in the proof of condition (iii) for $\mathcal{R}$, in particular Eqns. (\ref{eq:ineq})--(\ref{eq_chain_of_ineq_2}), implying that $V(\bar{x}_{i},\bar{x}^{\ast}_{i})  \leq \underline{\alpha}(\varepsilon)$ for any $i\in [1;\bar{N}]$; as a consequence $(\bar{x},\bar{x}^{\ast})\in\mathcal{R}$, hence one gets $(\bar{x}^{\ast},\bar{x})\in\mathcal{R}^{-1}$, concluding the proof.
\end{proof}

\begin{remark}
The symbolic models proposed in this section follow the work in \cite{PolaAutom2008,PolaSIAM2009,PolaSCL10,PolaTAC12,MajidTAC11}. In particular, the results of \cite{PolaSCL10} deal with symbolic models for nonlinear time--delay systems. We note that such results are not of help in the construction of symbolic models for NCS because they do not consider time--varying delays in the control input signals, which is one of the key features in NCS. 
\end{remark}

\section{Symbolic control design}\label{sec:control}
We consider a control design problem where the NCS $\Sigma$ has to satisfy a given specification robustly with respect to the non--idealities of the communication network. \\
The class of specifications that we consider is expressed by the (non--deterministic) transition system \cite{ModelChecking}: 
\begin{equation}
\mathcal{Q}=(X_{q},X_{q}^{0},\rTo_{q}),
\label{spec}
\end{equation}
where $X_{q}$ is a finite subset of $\mathbb{R}^{n}$, $X_{q}^{0}\subseteq X_{q}$ is the set of initial states and  $\rTo_{q}\subseteq X_{q}\times X_{q}$ is the transition relation. We suppose that $\mathcal{Q}$ is accessible, i.e. for any state $x\in X_{q}$ there exists a finite path from an initial condition $x_{0}\in X_{q}^{0}$ to $x$, i.e.
\[
x_{0} \rTo_{q} x_{1} \rTo_{q} x_{2} \rTo_{q} \,...\, \rTo_{q} x.
\]
Moreover we suppose that $\mathcal{Q}$ is non--blocking, i.e. for any $x\in X_{q}$ there exists $x'\in X_{q}$ such that $x \rTo_{q} x'$. 
For the subsequent developments we now reformulate the specification $\mathcal{Q}$ in the form of a system as in (\ref{def:system}), as follows:
\begin{equation}
Q^e=(X^e_q,X^{e,0}_q,U_q,\rTo_{e,q},Y^e_q,H^e_q),
\label{ext_spec}
\end{equation}
defined as follows:

\begin{itemize}
\item $X^e_q$ is the subset of $ X^0_q \cup \left( \bigcup_{N\in [N_{\min};N_{\max}]} X_q^{N} \right)$ such that for any $x=(x_{1},x_{2},...,x_{N})\in X^e_q$, with $N\in [N_{\min};N_{\max}]$, for any $i\in [1;N-1]$, the transition $x_{i} \rTo_q x_{i+1}$ is in $\mathcal{Q}$;

\item $X^{e,0}_q=X^0_q$;

\item $U_q=\{ \bar{u}_q \}$, where $\bar{u}_q$ is a \emph{dummy} symbol;

\item $x^{1}\rTo_{e,q}^{\bar{u}_q} x^{2}$, where:
\[
\begin{array}
{l}
\left\{
\begin{array}
{llll}
x^{1} & =(x^{1}_{1},x^{1}_{2},...,x^{1}_{N_{1}}), &  N_1 \in [N_{\min};N_{\max}];\\
x^{2} & =(x^{2}_{1},x^{2}_{2},...,x^{2}_{N_{1}}), &  N_2 \in [N_{\min};N_{\max}],
\end{array}
\right.
\end{array}
\]
and the transition $x^1_{N_1} \rTo_q x^2_{1}$ is in $\mathcal{Q}$;

\item $Y^e_q=X^e_q$;

\item $H^e_q=1_{X^e_q}$,

\end{itemize}

where $N_{\min}$ and $N_{\max}$ are as in (\ref{eq:N_min}) and (\ref{eq:N_max}). In order to cope with non-determinism in the communication network, symbolic controllers need to be robust in the sense of the following definition.
\begin{definition}
Given a system 
\[
S=(X_S,X_{S,0},U_S,\rTo_S,Y_S,H_S),
\]
a symbolic controller 
\[
C=(X_C,X_{C,0},U_C,\rTo_C,Y_C,H_C),
\]
is said to be \emph{robust with respect to $S$ with composition parameter $\theta\in\mathbb{R}^{+}$} if for any $u_s\in U_S$ and for each pair of transitions $x_s \rTo_S^{u_s} x_s^{\prime}$ and $x_s \rTo_S^{u_s} x_s^{\prime\prime}$ in $S$, with $x_s^{\prime}\neq x_s^{\prime\prime}$, the existence of a transition $(x_s,x_c)\rTo^{(u_s,u_c)} (x_s^{\prime},x_c^{\prime})$ in $S \Vert_{\theta} C$, for some $x_c,x_c^{\prime}\in X_C$, implies the existence of a transition $(x_s,x_c)\rTo^{(u_s,u_c)} (x_s^{\prime\prime},x_c^{\prime\prime})$ in $S \Vert_{\theta} C$ for some $x_c^{\prime\prime}\in X_C$.
\label{def:rob_controller}
\end{definition}

We are now ready to state the control problem that we address in this section.

\begin{problem}
\label{problem}
Consider the NCS $\Sigma$, the specification $Q^e$ in (\ref{ext_spec}) and a desired precision $\varepsilon\in\mathbb{R}^{+}$. Find a parameter $\theta\in\mathbb{R}^{+}$ and a symbolic controller $C$ such that:
\begin{itemize}
\item [(1)] $C$ is robust with respect to $S(\Sigma)$ with composition parameter $\theta$;
\item [(2)] $S(\Sigma)\Vert_{\theta} C  \preceq_{\varepsilon} Q^e$;
\item [(3)] $S(\Sigma)\Vert_{\theta} C$ is non--blocking.
\end{itemize}
\end{problem}

Condition (1) of Problem \ref{problem} is posed to cope with the non-determinism of $S(\Sigma)$. The approximate similarity inclusion in (2) requires the state trajectories of the NCS to be close to the ones of specification $Q^e$ up to the accuracy $\varepsilon$. The non-blocking condition (3) prevents deadlocks in the interaction between the plant and the controller. 

In the following definition, we provide the controller $C^{\ast}$ that will be shown to solve Problem \ref{problem}.

\begin{definition}
The symbolic controller $C^{\ast}$ is the maximal sub--system\footnote{Here maximality is defined with respect to the preorder induced by the notion of sub--system.} $C$ of $S_{\ast}(\Sigma) \Vert_{\mu_{x}} Q^e$ that satisfies the following properties:
\begin{itemize}
\item $C$ is non--blocking;
\item for any $u_{\ast}\in U_{\ast}$ and for each pair of transitions $x \rTo_{\ast}^{u_{\ast}} x^{\prime}$ and $x \rTo_{\ast}^{u_{\ast}} x^{\prime\prime}$ in $S_{\ast}(\Sigma)$, with $x^{\prime}\neq x^{\prime\prime}$, the existence of a transition $(x,x_q)\rTo^{(u_{\ast},\bar{u}_q)} (x^{\prime},x_q^{\prime})$ in $C$, for some $x_q$, $x_q^{\prime}$, implies the existence of a transition $(x,x_q)\rTo^{(u_{\ast},\bar{u}_q)}$ $(x^{\prime\prime},x_q^{\prime\prime})$ in $C$, for some $x_q^{\prime\prime}$.
\end{itemize}
\label{canon_contr}
\end{definition}

The following technical result will be useful in the sequel.

\begin{lemma}
Let \mbox{$S_{i}=(X_{i},X_{0,i},U_{i},\rTo_{i},Y_{i},H_{i})$} ($i=1$, $2$, $3$) be metric systems with the same output sets $Y_{1}=Y_{2}=Y_{3}$ and metric $d$. Then the following statements hold:
\begin{itemize}
\item[(i)] \cite{AB-TAC07} for any $\varepsilon_{1}\leq \varepsilon_{2}$, $S_{1}\preceq_{\varepsilon_{1}}S_{2}$ implies \mbox{$S_{1}\preceq_{\varepsilon_{2}}S_{2}$};
\item[(ii)] \cite{AB-TAC07} if $S_{1}\preceq_{\varepsilon_{12}}S_{2}$ and $S_{2}\preceq_{\varepsilon_{23}}S_{3}$ then $S_{1}\preceq_{\varepsilon_{12}+\varepsilon_{23}}S_{3}$;
\item[(iii)] \cite{PolaTAC12} for any $\theta\in\mathbb{R}^{+}_{0}$, $S_{1} \Vert_{\theta} S_{2} \preceq_{\theta} S_{2}$.
\end{itemize}
\label{lemma1}
\end{lemma}

We are now ready to show that the controller $C^{\ast}$ solves Problem \ref{problem}. 

\begin{theorem}
Consider the NCS $\Sigma$ and the specification $Q^e$. Suppose that the control system $P$ in $\Sigma$ enjoys Assumptions (H1) and (H2) in Theorem \ref{polaut}. Then for any desired precision $\varepsilon\in\mathbb{R}^{+}$ and for any $\theta,\mu_{x}\in\mathbb{R}^{+}$ such that:
\begin{align}
& \mu_{x}+\theta  \leq\varepsilon,\label{solving_condition_1}\\
& \mu_{x} \leq \min \left \{\gamma^{-1} \left ( \left ( 1-e^{-\lambda \tau} \right ) \underline{\alpha}(\theta)  \right ) , \overline{\alpha}^{-1}(\underline{\alpha}(\theta)), \hat{\mu}_{X} \right \}\text{,}\label{solving_condition_2}
\end{align}
the symbolic controller $C^{\ast}$ solves Problem \ref{problem}.
\label{Mmain}
\end{theorem}

\begin{proof}
First we prove condition (1) of Problem \ref{problem}. Consider any $u\in U_{\tau}$, any state $x\in X_{\tau}$, and any pair of transitions $x \rTo^{u}_{\tau} x^{\prime}$ and $x \rTo^{u}_{\tau} x^{\prime\prime}$ in $S(\Sigma)$, with $x^{\prime}\neq x^{\prime\prime}$. Consider any transition $(x,x_c)\rTo^{(u,u_c)}(x^{\prime},x^{\prime}_c)$ in $S(\Sigma)\Vert_{\theta} C^{\ast}$, where $x_c=(x_{\ast},x_q)$, $x^{\prime}_c=(x^{\prime}_{\ast},x^{\prime}_q)$, $u_c=(u_{\ast},\bar{u}_q)$, since $C^{\ast}\sqsubseteq S_{\ast}(\Sigma) \Vert_{\mu_{x}} Q^e$. Note that the transition $x_c\rTo^{u_c}x^{\prime}_c$ (equivalently $(x_{\ast},x_q)\rTo^{(u_{\ast},\bar{u}_q)}(x^{\prime}_{\ast},x^{\prime}_q)$) is in $C^{\ast}$ by Definition \ref{composition}. By definition of $S(\Sigma)$ and $S_{\ast}(\Sigma)$ and in view of condition (\ref{solving_condition_2}) and Assumptions (H1)-(H2) in Theorem \ref{polaut}, ensuring that $S(\Sigma)\cong_{\theta}^{\alt}S_{\ast}(\Sigma)$, the existence of a transition $x \rTo^{u}_{\tau} x^{\prime\prime}$ in $S(\Sigma)$ implies the existence of a transition $x_{\ast} \rTo^{u_{\ast}}_{\ast} x_{\ast}^{\prime\prime}$ in $S_{\ast}(\Sigma)$ s.t. $d_{Y_{\tau}}(x^{\prime\prime},x_{\ast}^{\prime\prime})\leq \theta$, with $x_{\ast}^{\prime\prime} \neq x_{\ast}^{\prime}$, in general. Furthermore, by Definition \ref{canon_contr}, the existence of the transitions $x_{\ast} \rTo^{u_{\ast}}_{\ast} x_{\ast}^{\prime}$ and $x_{\ast} \rTo^{u_{\ast}}_{\ast} x_{\ast}^{\prime\prime}$ in $S_{\ast}(\Sigma)$ and of the transition $(x_{\ast},x_q)\rTo^{(u_{\ast},\bar{u}_q)}(x^{\prime}_{\ast},x^{\prime}_q)$ in $C^{\ast}$ implies the existence of a transition $(x_{\ast},x_q)\rTo^{(u_{\ast},\bar{u}_q)}(x^{\prime\prime}_{\ast},x^{\prime\prime}_q)$ in $C^{\ast}$ for some $x^{\prime\prime}_q$. Since $d_{Y_{\tau}}(x^{\prime\prime},x_{\ast}^{\prime\prime})\leq \theta$, the transition $(x,x_c)\rTo^{(u,u_c)}(x^{\prime\prime},x^{\prime\prime}_c)$, with $x^{\prime\prime}_c=(x^{\prime\prime}_{\ast},x^{\prime\prime}_q)$,  is in $S(\Sigma)\Vert_{\theta} C^{\ast}$, which concludes the proof of condition (1) of Problem \ref{problem}.

We now show that condition (2) of Problem \ref{problem} is fulfilled. By Lemma \ref{lemma1} (iii), $S(\Sigma)\Vert_{\theta} C^{\ast}\preceq_{\theta} C^{\ast}$ and $S_{\ast}(\Sigma)\Vert_{\mu_x} Q^e \preceq_{\mu_x} Q^e$. Since $C^{\ast}$ is a sub--system of $S_{\ast}(\Sigma) \Vert_{\mu_{x}} Q^e$ then $C^{\ast}\preceq_{0} S_{\ast}(\Sigma) \Vert_{\mu_{x}} Q^e$. By Lemma \ref{lemma1} (i)-(ii), and from (\ref{solving_condition_1}), $\mu_{x}+\theta  \leq\varepsilon$, the above approximate similarity inclusions imply $S(\Sigma)\Vert_{\theta} C^{\ast} \preceq_{\varepsilon} Q^e$, which concludes the proof of condition (2) of Problem \ref{problem}. 

We finally show that also condition (3) holds. Consider any state $(x,x_{\ast},x_q)$ of $S(\Sigma)\Vert_{\theta} C^{\ast}$. Since $C^{\ast}$ is non--blocking, for the state $(x_{\ast},x_q)$ of $C^{\ast}$ there exists a state $(x^{\prime}_{\ast},x^{\prime}_q)$ of $C^{\ast}$ such that $(x_{\ast},x_q)\rTo^{(u_{\ast},\bar{u}_q)}(x^{\prime}_{\ast},x^{\prime}_q)$ is a transition of $C^{\ast}$ for some $(u_{\ast},\bar{u}_q)$. Since by the inequality in (\ref{solving_condition_2}) and Theorem \ref{polaut}, $S(\Sigma)$ and $S_{\ast}(\Sigma)$ are $A \theta A$--bisimilar, for the transition $x_{\ast}\rTo^{u_{\ast}} x_{\ast}^{\prime}$ in $S_{\ast}(\Sigma)$ there exists a transition $x\rTo^{u} x^{\prime}$ in $S(\Sigma)$ such that $d_{Y_{\tau}}(x^{\prime},x_{\ast}^{\prime})\leq \theta$. This implies from Definition \ref{composition} that $(x^{\prime},x^{\prime}_{\ast},x^{\prime}_q)$ is a state of $S(\Sigma)\Vert_{\theta} C^{\ast}$ and therefore that \\ $(x,x_{\ast},x_q)\rTo^{(u,u_{\ast},\bar{u}_q)}(x^{\prime},x^{\prime}_{\ast},x^{\prime}_q)$ is a transition of $S(\Sigma)\Vert_{\theta} C^{\ast}$, which concludes the proof.

\end{proof}

\section{An illustrative example}\label{sec:example}
We consider a pair of nonlinear control systems $P_a$ and $P_b$ described by the following differential equations:  
\begin{align}
\dot{x} &=\left[
\begin{array}
[l]{l}
\dot{x}_1\\
\dot{x}_2
\end{array}
\right] =&f(x,u) &=
\left[
\begin{array}
[l]{l}
x_2\\
-5 \sin(x_1)-4 x_2+ u
\end{array}
\right],\label{example1}\\
\dot{z} &=\left[
\begin{array}
[l]{l}
\dot{z}_1\\
\dot{z}_2
\end{array}
\right] =&g(z,v) &=
\left[
\begin{array}
[l]{l}
-2.5 z_1+z_2^2\\
2 z_1-6 e^{z_2}+ v + 6
\end{array}
\right],
\label{example2}
\end{align}

where $x\in X=X_{0} =\left[-\frac{\pi}{3},\frac{\pi}{3}\right[\times [-1,1[$, $u\in U =[-5,5]$, $z\in Z=Z_{0} =[-1,1[\times [-1,1[$ and $v\in V  =[-5,5]$. The two plants that are denoted by $\Sigma_a$ and $\Sigma_b$, form a pair of NCS loops as the one depicted in Figure \ref{NCSpic}. 
The two controllers are supposed to run on a shared CPU that is able to control both processes. The shared network/computation parameters are $B_{\max} =1 \text{ kbit}/s$, $\tau =0.2 s$, $\Delta_{\min}^{\ctrl}=0.001 s$, $\Delta_{\max}^{\ctrl} =0.01 s$ and $\Delta^{\req}_{\max} =0.1 s$.
The output quantization is chosen to be equal to $\mu_{x}=2\cdot 10^{-4}$ for both the NCS, while we set a different input quantization: $\mu_{u} =0.0024$ for $\Sigma_a$ and $\mu_{u} =2\cdot 10^{-4}$ for $\Sigma_b$. We assume that $P_b$ is farther away than $P_a$ (in terms of hops in the network topology) from the shared CPU, resulting in larger delays; in particular, we set $\Delta_{\min}^{\delay,a}=0.05 s$, $\Delta_{\max}^{\delay,a}=0.12 s$ for $\Sigma_a$ and $\Delta_{\min}^{\delay,b}=0.1 s$, $\Delta_{\max}^{\delay,b}=0.24 s$ for $\Sigma_b$. As from Eqns. (\ref{eq:N_min})-(\ref{eq:delta_max}), this results in $N^a_{\min}=1$, $N^a_{\max}=3$ for $\Sigma_a$, and $N^b_{\min}=2$, $N^b_{\max}=4$ for $\Sigma_b$.
We consider the following common quadratic Lyapunov function:
\begin{align*}
V(y,y^{\prime})=\frac{1}{2}\Vert y-y^{\prime}\Vert_{2}^{2},
\end{align*}
satisfying condition (i) of Definition \ref{dGAS_Lyapunov} with $\underline{\alpha}(r)=0.5\, r^2$ and $\overline{\alpha}(r)=r^2$, $r\in\mathbb{R}^{+}_{0}$. Furthermore, for the first control system $P_a$, one can write:
\begin{align*}
\frac{\partial V}{\partial x}f(x,u)
+\frac{\partial V}{\partial x^{\prime}}f(x^{\prime},u) & =(x-x^{\prime})^{T}(f(x,u)-f(x^{\prime},u))  =\\
&  \leq -0.75V(x,x^{\prime}).
\end{align*}
Condition (ii) of Definition \ref{dGAS_Lyapunov} is therefore fulfilled for $P_a$ with $\lambda_a=0.75$. Analogous computation for $P_b$ leads to $\lambda_b=0.2$. Hence, by Theorem \ref{TH-IGAS}, control systems (\ref{example1}) and (\ref{example2}) are $\delta$--GAS. 
In order to construct symbolic models for $\Sigma_a$ and $\Sigma_b$, we apply Theorem \ref{thmain}. Assumption (H1) holds by the incremental stability property proven above. Assumption (H2) of Theorem \ref{TH-IGAS} holds with $\gamma(r)=2.09r$ for $P_a$ and $\gamma(r)=2r$ for $P_b$. Finally, for a precision $\varepsilon_a=\pi/20$ and $\varepsilon_b=0.2$ for $\Sigma_a$ and $\Sigma_b$, respectively, the inequality in (\ref{bisim_condition1}) holds. Hence, we can construct symbolic models for $S_{\ast}(\Sigma_a)$ and $S_{\ast}(\Sigma_b)$ that are $A\varepsilon_a A$ bisimilar and $A\varepsilon_b A$ bisimilar to $S(\Sigma_a)$ and $S(\Sigma_b)$. For $S_{\ast}(\Sigma_a)$, the resulting number of states is $1.8\cdot 10^{22}$ and the number of control inputs is $2,049$; $S_{\ast}(\Sigma_b)$ instead contains $3.91\cdot 10^{29}$ states and $16,385$ control inputs.
Due to the large size of the symbolic models obtained, further details are not included here. 
We now use the results in Section \ref{sec:control} to solve trajectory tracking problems (on a finite time horizon), expressed in the form of Problem \ref{problem}. We consider specifications expressed in the form of transition systems $\mathcal{Q}_a$ and $\mathcal{Q}_b$, as in (\ref{spec}). The specification $\mathcal{Q}_a$ is given by the following trajectory on the first state variable:
\begin{equation*}
\begin{array}
{l}
0.5 \rTo 0.4 \rTo 0.3 \rTo 0.2 \rTo 0.1 \rTo \\
0 \rTo -0.2 \rTo -0.35 \rTo -0.5 \rTo \\
-0.6\rTo -0.7\rTo -0.8 \rTo -0.8\rTo \\
-0.75\rTo -0.7,
\end{array}
\label{spec1}
\end{equation*}
while the specification $\mathcal{Q}_b$ is given by the following trajectory:
\begin{equation*}
\begin{array}
{l}
(0.5,0.5) \rTo(0.4,0.3) \rTo (0.3,0.2)\rTo (0.2,0.1)\rTo \\
(0.1,-0.1)\rTo(0,-0.25) \rTo(-0.1,-0.3)\rTo \\
(-0.1,-0.4)\rTo(-0.15,-0.4)\rTo (-0.15,-0.4)\rTo\\
(0.1,-0.3)\rTo(0.2,-0.2)\rTo(0.2,-0.1)\rTo\\
(0.2,-0.1)\rTo(0.2,-0.05).
\end{array}
\label{spec2}
\end{equation*}
For the choice of the interconnection parameter $\theta_a=0.9 \varepsilon_a$ and $\theta_b=0.9 \varepsilon_b$, for the two NCS loops, Theorem \ref{Mmain} holds and a controller $C^{\ast}$ as from Definition (\ref{canon_contr}) solves the control problem. Since the symbolic models of $\Sigma_a$ and $\Sigma_b$ have large size, a straightforward application of the results reported in the previous section for the design of the requested symbolic controllers would exhibit a large space and time computational complexity. For this reason in this example we adapt to NCS the algorithms proposed in \cite{PolaTAC12} concerning the integrated symbolic control design of nonlinear control systems. More precisely, instead of first computing the symbolic models of the plants to then derive the symbolic controllers, we integrate the design of the symbolic controllers with the construction of the symbolic models. By using this approach we designed the requested symbolic controllers in $2,039 s$ with a total memory occupation of $25,239$ integers; this computation has been performed on the Matlab suite through an Intel Core 2 Duo T5500 1.66GHz laptop with 4 GB RAM. 
The synthesized controllers has been validated through the OMNeT++ network simulation framework \cite{Omnet}. Communication delays are managed in OMNeT++ by means of a variable number of hops for each message and random delays over each network hop. We set a delay over the single hop variable between $0.0125s$ and $0.02 s$, and a number of network hops between $4$ and $6$ for $\Sigma_a$ and between $8$ and $12$ for $\Sigma_b$. Figure \ref{fig:omnet_scheme} shows the OMNeT++ implementation of the two-loop network scheme with shared CPU.
In Figures \ref{fig:sim1} and \ref{fig:sim2}, we show the simulation results for the tracking problems considered, for a particular realization of the network uncertainties: it is easy to see that the specifications are indeed met.

\begin{figure}[ht]
\begin{center}
\includegraphics[scale=0.1]{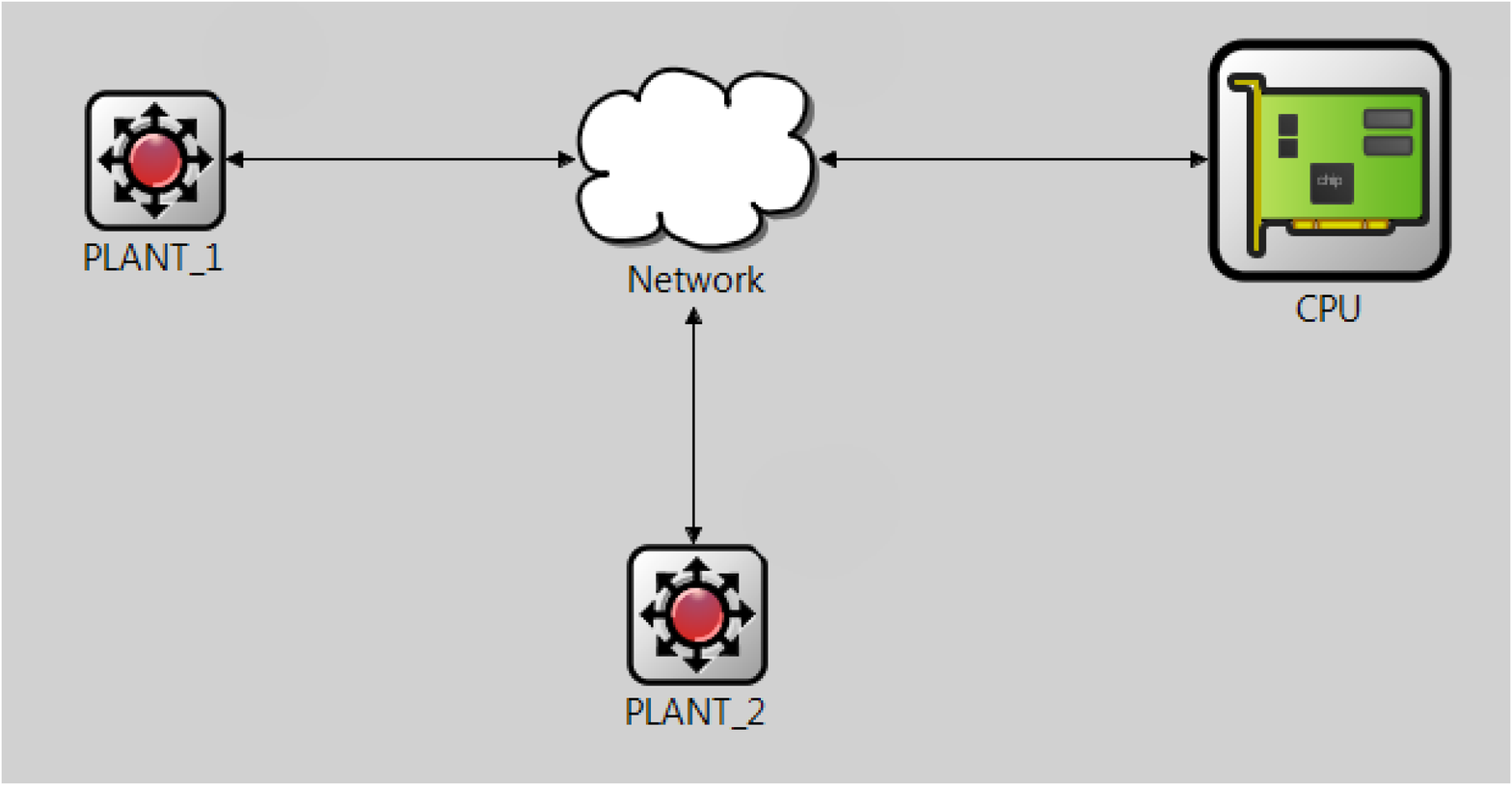}
\caption{OMNeT++ implementation of Networked Control Systems with Symbolic Controller.} 
\label{fig:omnet_scheme}
\end{center}
\end{figure}

\begin{figure}[ht]
\begin{center}
\subfigure{\includegraphics[scale=0.1]{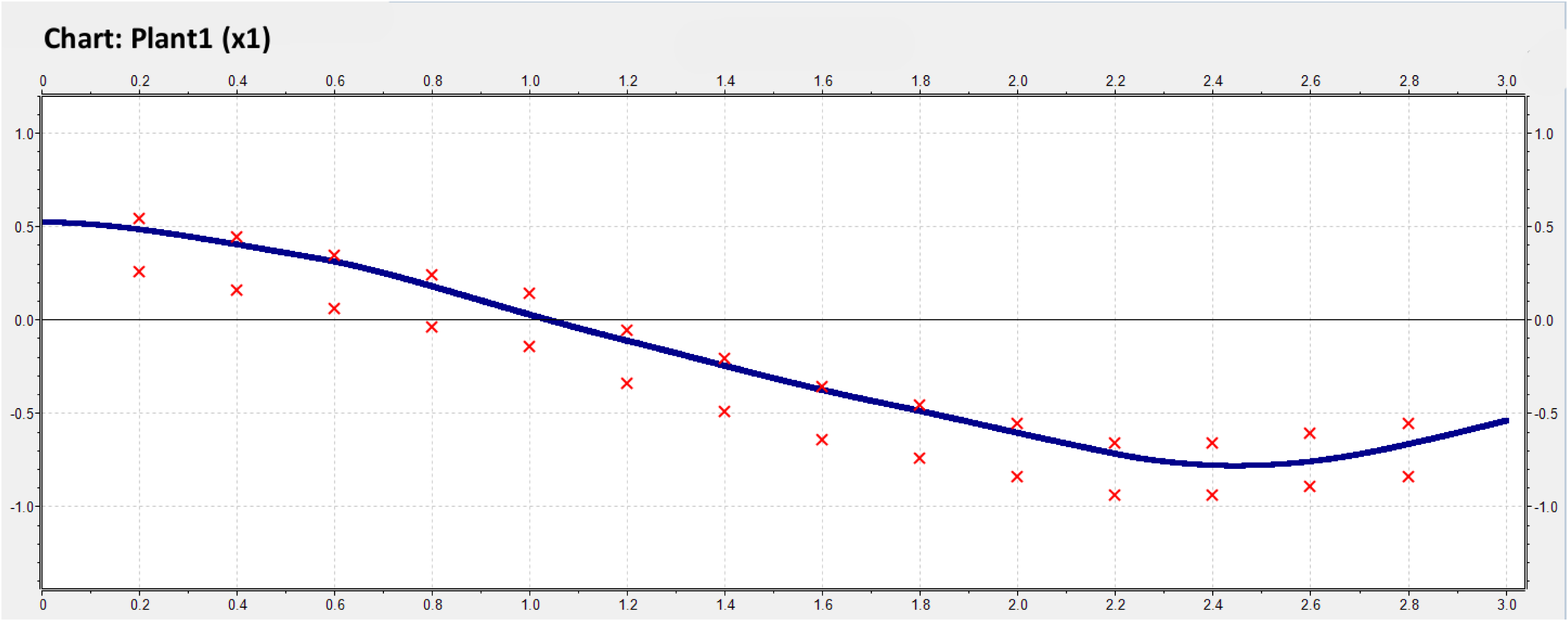}}
\subfigure{\includegraphics[scale=0.1]{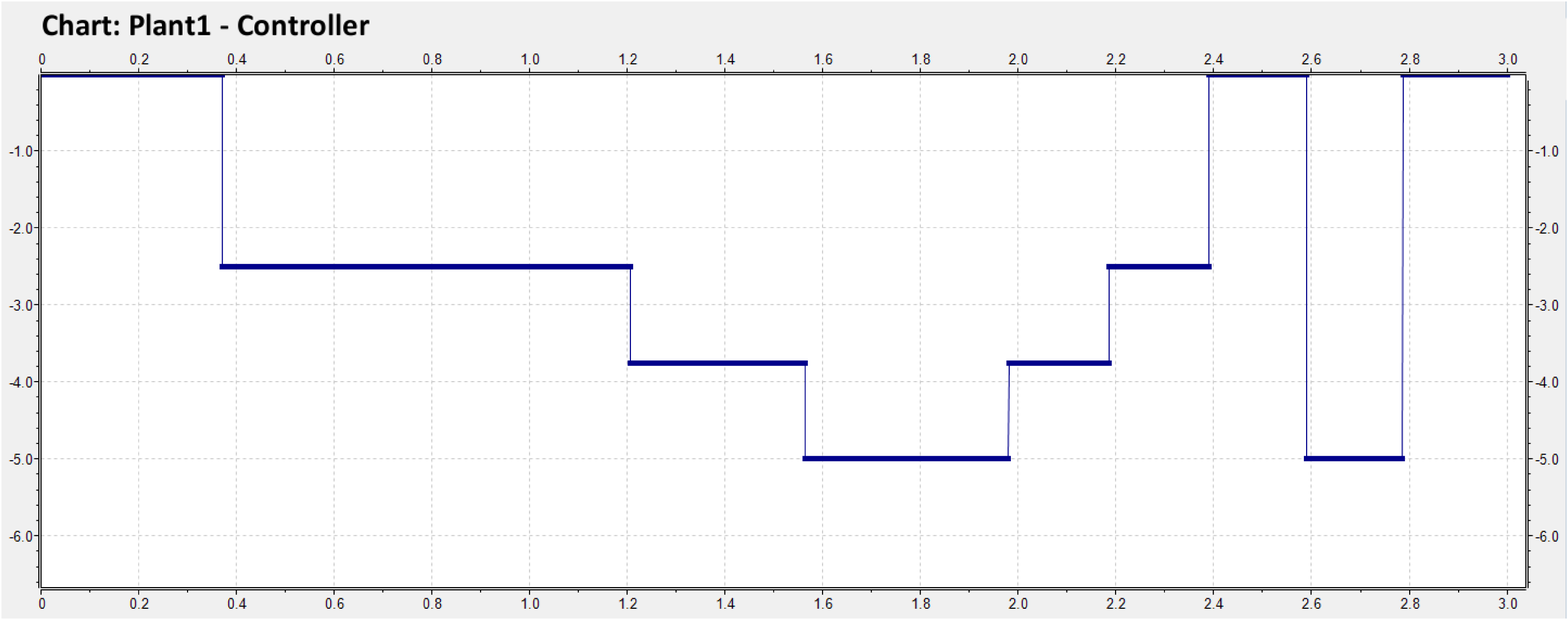}}
\caption{State trajectory and control input for the NCS $\Sigma_a$.} 
\label{fig:sim1}
\end{center}
\end{figure}

\begin{figure}[ht]
\begin{center}
\subfigure{\includegraphics[scale=0.1]{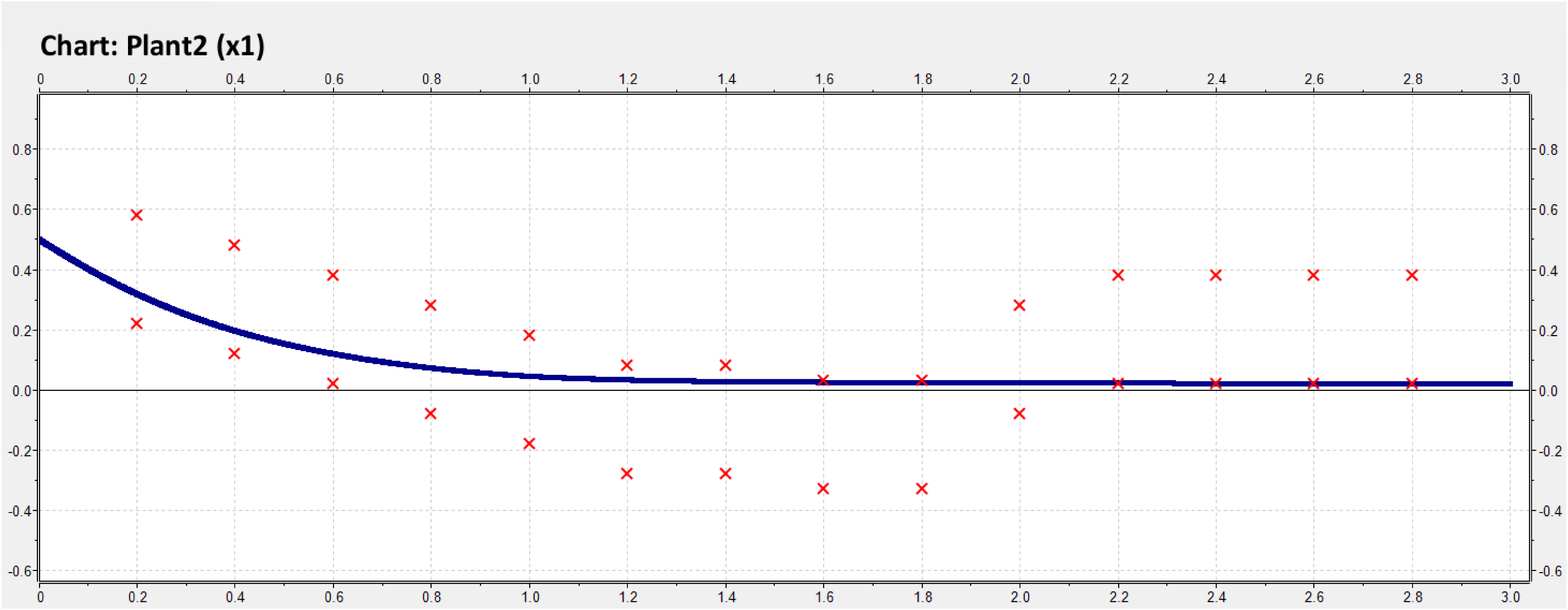}}
\subfigure{\includegraphics[scale=0.1]{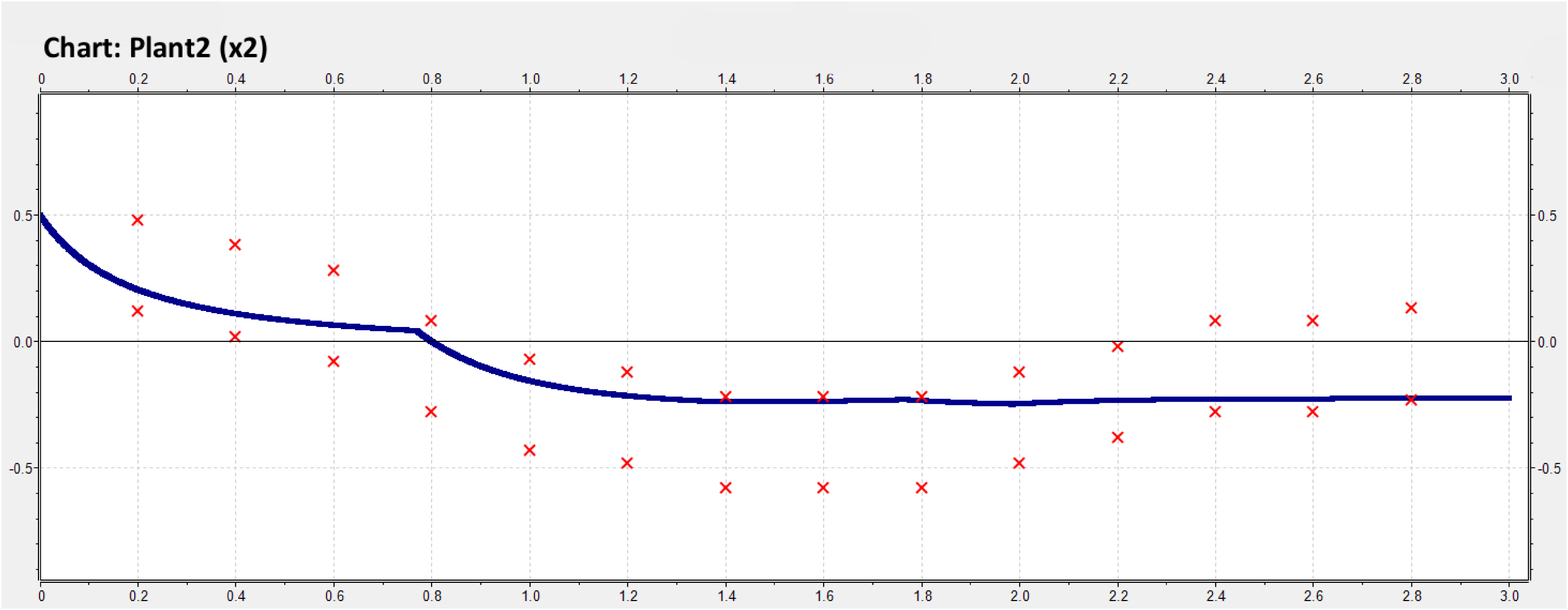}}
\subfigure{\includegraphics[scale=0.1]{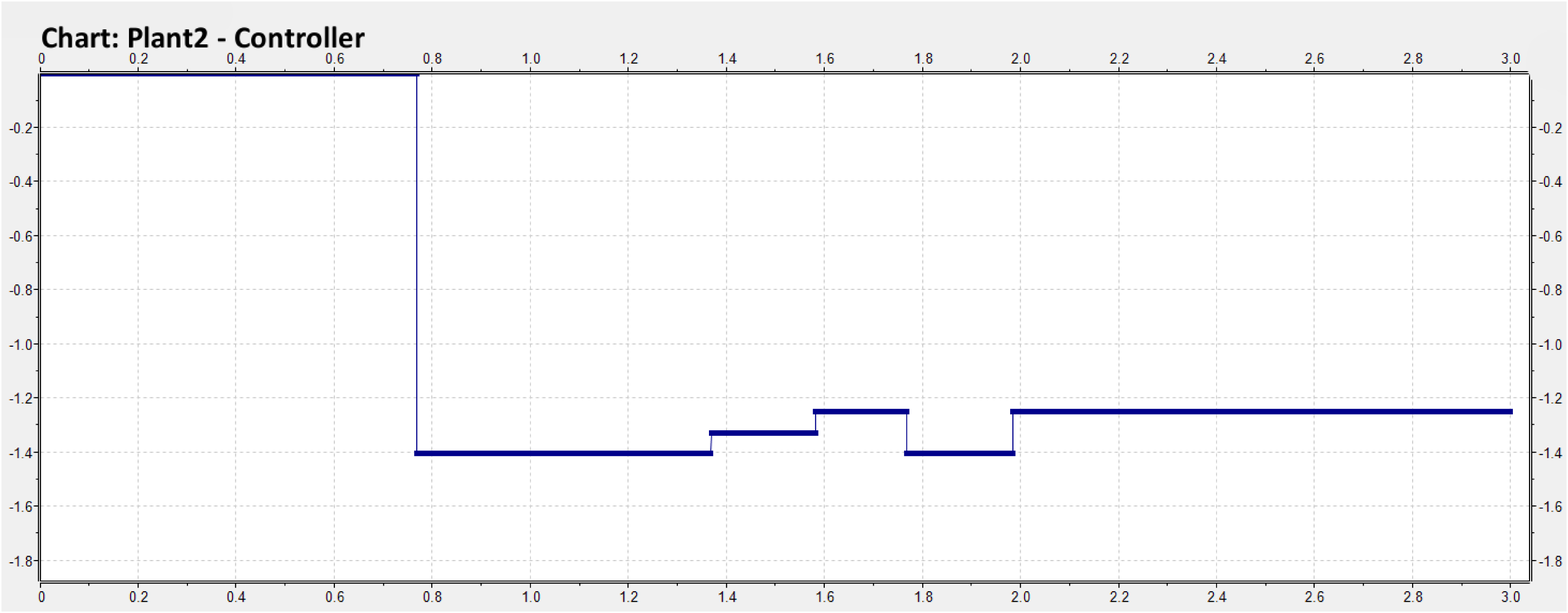}}
\caption{State trajectory and control input for the NCS $\Sigma_b$.} 
\label{fig:sim2}
\end{center}
\end{figure}

\section{Conclusions}\label{sec:conclusion}
In this paper we proposed a symbolic approach to the control design of nonlinear NCS. Under the assumption of $\delta$--GAS, symbolic models were proposed, which approximate NCS in the sense of alternating approximate bisimulation. These symbolic models were used to solve symbolic control problems on NCS where specifications are expressed in terms of automata on infinite strings. The assumption of $\delta$--GAS in the plant control system of the NCS is a key ingredient in our results because if a digital controller is found which enforces the desired specification on the symbolic model, the notion of alternating approximate bisimulation guarantees that the specification is fulfilled on the NCS within a given accuracy that can be chosen as small as desired. Conversely if a control strategy solving the control problem does not exist, the notion of alternating approximate bisimulation guarantees that such a solution does not exist on the original NCS. If compared with existing results on NCS, the main drawback of the proposed results is in the assumption of incremental stability on the plant control systems. 
One way to overcome this crucial assumption is to leverage the results reported in \cite{MajidTAC11}, which propose symbolic models approximating (possibly) unstable nonlinear control systems in the sense of alternating approximate simulation. This point is under investigation.

\section*{Acknowledgments}
We are grateful to Pierdomenico Pepe for fruitful discussions on the topics of this paper and to
Daniele De Gregorio and Quirino Lo Russo for the implementation of the example proposed in Section 7 in the OMNeT++  network simulation framework.

\bibliographystyle{abbrv}
\bibliography{biblio1}  

\end{document}